\documentclass{article}

\usepackage{spconf}

\usepackage{amsmath, amsfonts, amsthm}
\usepackage{booktabs}

\usepackage{doi}
\usepackage{enumitem}

\usepackage{hyperref}
\hypersetup{
  colorlinks = true,
  allcolors = blue,
}
\usepackage[capitalise,nameinlink]{cleveref}
\usepackage{siunitx}
\usepackage{nicefrac}
\usepackage{nicematrix}

\usepackage{graphicx}
\usepackage{tikz}
\usetikzlibrary{matrix,arrows,shapes,positioning,fit,calc,decorations.pathmorphing,intersections}

\usepackage[outline]{contour}
\contourlength{4.0pt}

\usepackage{pgfplots}
\pgfplotsset{compat=1.17}
\usepgfplotslibrary{groupplots}

\usepackage{xcolor}
\colorlet{myLightGray}{white!90!black}
\colorlet{myMediumGray}{white!50!black}
\pgfdeclareverticalshading{coolwarm}{100bp}{color(0bp)=(blue); color(50bp)=(myLightGray); color(100bp)=(red)}
\pgfdeclareverticalshading{warm}{100bp}{color(0bp)=(white); color(100bp)=(red)}

\pgfdeclarelayer{background}
\pgfdeclarelayer{fg}    
\pgfdeclarelayer{ffg}   
\pgfsetlayers{background,main,fg,ffg}  

\newtheorem{prop}{Proposition}

\newtheoremstyle{myremark}
{\topsep} 
{\topsep} 
{\normalfont} 
{} 
{\bfseries} 
{.} 
{5pt plus 1pt minus 1pt} 
{\thmname{#1}\thmnumber{ #2}\thmnote{ (#3)}} 

\theoremstyle{myremark}

\theoremstyle{definition}

\crefmultiformat{prop}{Propositions~#2#1#3}{ and~#2#1#3}{, #2#1#3}{ and~#2#1#3}
\crefmultiformat{theorem}{Theorems~#2#1#3}{ and~#2#1#3}{, #2#1#3}{ and~#2#1#3}
\crefmultiformat{lemma}{Lemmas~#2#1#3}{ and~#2#1#3}{, #2#1#3}{ and~#2#1#3}
\crefmultiformat{coro}{Corollaries~#2#1#3}{ and~#2#1#3}{, #2#1#3}{ and~#2#1#3}
\crefmultiformat{problem}{Problems~#2#1#3}{ and~#2#1#3}{, #2#1#3}{ and~#2#1#3}
\crefmultiformat{assump}{Assumptions~#2#1#3}{ and~#2#1#3}{, #2#1#3}{ and~#2#1#3}

\crefmultiformat{defn}{Definitions~#2#1#3}{ and~#2#1#3}{, #2#1#3}{ and~#2#1#3}
\crefrangeformat{defn}{Definitions.~#3#1#4,#5#2#6}

\crefmultiformat{property}{Properties~#2#1#3}{ and~#2#1#3}{, #2#1#3}{ and~#2#1#3}
\crefrangeformat{property}{Properties~#3#1#4,#5#2#6}

\newcommand{\bx}{\ensuremath{\mathbf{x}}}

\newcommand{\bA}{\ensuremath{\mathbf{A}}}
\newcommand{\bB}{\ensuremath{\mathbf{B}}}
\newcommand{\bF}{\ensuremath{\mathbf{F}}}
\newcommand{\bL}{\ensuremath{\mathbf{L}}}

\newcommand{\bbR}{\ensuremath{\mathbb{R}}}

\newcommand{\cC}{\ensuremath{\mathcal{C}}}

\newcommand{\cS}{\ensuremath{\mathcal{S}}}

\newcommand{\del}{\ensuremath{\partial}}

\newcommand{\id}{\ensuremath{\mathrm{id}}}

\DeclareMathOperator*{\argmin}{arg\,min}

\newcommand{\nohyphen}{\mbox{-}\nobreak\hspace{0pt}}
\newcommand{\subheading}[1]{\smallskip\noindent\textbf{#1.} }

\begin{document}

\ninept

\title{Signal Processing on Product Spaces}

\name{T. Mitchell Roddenberry$^{\star}$,
Vincent P. Grande$^{\dagger}$,
	Florian Frantzen$^{\dagger}$,
	Michael T. Schaub$^{\dagger}$,
	Santiago Segarra$^{\star}$
        \thanks{TMR and SS were supported by USA NSF under award CCF-2008555. FF, VG and MTS acknowledge partial funding from the Ministry of Culture and Science of North Rhine-Westphalia (NRW R\"{u}ckkehrprogramm), the Excellence Strategy of the Federal Government and the L\"{a}nder, and the Deutsche Forschungsgemeinschaft (DFG, German Research Foundation) – 2236/2. }
      }

\address{$^{\star}$Department of Electrical and Computer Engineering, Rice University, USA\\
$^{\dagger}$Department of Computer Science, RWTH Aachen University, Germany}

\maketitle

\begin{abstract}
    We establish a framework for signal processing on product spaces of simplicial and cellular complexes.
    For simplicity, we focus on the product of two complexes representing time and space, although our results generalize naturally to products of simplicial complexes of arbitrary dimension.
    Our framework leverages the structure of the eigenmodes of the Hodge Laplacian of the product space to jointly filter along time and space.
    To this end, we provide a decomposition theorem of the Hodge Laplacian of the product space, which highlights how the product structure induces a decomposition of each eigenmode into a spatial and temporal component.
    Finally, we apply our method to real world data, specifically for interpolating trajectories of buoys in the ocean from a limited set of observed trajectories.
\end{abstract}

\section{Introduction}

In recent years, graph signal processing (GSP) has proven itself to be an important tool for analyzing, filtering, denoising, and interpolating signals defined on abstract domains. 
Extending the tools of ordinary signal processing to graphs opened the door for a multitude of more combinatorially inspired applications \cite{Shuman:2013, Sandryhaila:2013,  Scholkemper:2022, Perraudin:2014}.
Recently, this shift has been extended further to signal processing on topological spaces using simplicial complexes (SCs) \cite{Schaub:2021, Schaub:2022, Barbarossa:2020, Barbarossa:2018, Yang:2021}, cellular complexes (CCs) \cite{Roddenberry:2021, Robinson:2014, Sardellitti:2021}, or cellular sheaves \cite{Bodnar:2022}.
This move to topological signal processing allows both for a more general geometry of the underlying domain to be captured, as well as for additional constraints on the filtering to be taken into account. 
Example applications where this topological viewpoint has been used 
include wireless traffic networks \cite{Barbarossa:2020} or outlier detection on whales traveling in the Canadian Arctic Archipelago \cite{Frantzen:2021}.

Sometimes the data considered varies not only in spatial dimensions, but also has a temporal dimension.
An effective filter then should consider both of these dimensions.
However, filtering on the simplicial complex will only take the spatial dimension into account, whereas filtering individually along each edge in time disregards the underlying topology and geometry of the problem.
To solve this issue we propose to construct a new complex that captures both the spatial and temporal dimensions of the problem.

The idea of signal processing for product spaces has already been considered for product graphs for data defined on the nodes \cite{Ortiz:2018, Sandryhaila:2014, Shi:2020,grassi2017time}.
However, processing of edge flow data is not possible within these product graph frameworks: while the product of two $0$-dimensional vertices is again a $(0+0)$-dimensional vertex and can be considered via standard graph products, the product of two ($1$-dimensional) edges is not an edge again but a $2$-dimensional face.
Hence, more sophisticated techniques and models are needed to model and process such data.

\subsection{Contributions and Outline}

We present a novel framework for signal processing on product spaces of simplicial complexes, which enables us to filter time-varying signals on SCs.
By interpreting the time axis as an SC (in particular, as a graph), we construct a cellular complex that represents the spatially distributed, time-varying signal.
It turns out that the Hodge Laplacian for the resulting product complex is a concise way to express a filter with the desired properties.

\Cref{section:background} gives an overview of the mathematical concepts of simplicial and cellular complexes, the Hodge Laplacian, and product complexes.
\Cref{section:motivation} elaborates on the motivation and the theory for temporal flow interpolation.
Specifically, \Cref{paragraph:HodgeLaplacian} describes the Hodge Laplacian of a product space.
Finally, we apply our method to real world ocean currents and show that interpolating along both time and space using the Hodge Laplacian of the product space outperforms interpolation solely along a single dimension.

\section{Background}

\label{section:background}
In this section, we present an elementary overview of concepts used to process signals defined on higher order networks such as simplicial or cell complexes.
For more details, \cite{Hatcher:2005, Bredon:1993} give a background in algebraic topology and \cite{Roddenberry:2021, Schaub:2021} in topological signal processing.

\subheading{Notation}
For some $n \in \mathbb{N}$, let $[n]$ denote the set of integers $\{ 1, \dots, n \}$.
Matrices and vectors are denoted by letters $\bA, \bx$.
We use $\cong$ to denote two isomorphic vector spaces, and $\oplus$ to denote the direct sum of vector spaces.
Further, $\otimes$ denotes the tensor product.

\subheading{Simplicial Complexes}
A simplicial complex \cite{Hatcher:2005,Bredon:1993} (SC) $X$ consists of a finite set of points $\mathcal{V}$, and a set of nonempty subsets of $\mathcal{V}$ that is closed under taking nontrivial subsets.
A \emph{$k$-simplex} $\cS^k$ is a subset of $\mathcal{V}$ with $k+1$ points.
Further, the \emph{faces} of $\cS^k$ are the subsets of $\cS^k$ with cardinality $k$, that is, exactly one element of the simplex is omitted.
If $\cS^{k-1}$ is a face of $\cS^k$, $\cS^k$ is called a \emph{coface} of $\cS^{k-1}$.
We denote the set of $k$-simplices and their number in $X$ by $X_k$ and $N_k$, respectively.
We can ground the understanding of SCs as an extension of the well known concept of graphs: $X_0$ is the set of vertices, $X_1$ is the set of edges, $X_2$ is the set of filled in triangles (not all 3-cliques must be filled in triangles), and so on.
The orientation of a $k$-simplex $\cS^k$ is represented by an ordering of its $k+1$ elements of $\mathcal{V}$. Two orientations are considered the same if they differ by an even permutation. Thus every simplex admits two possible orientations.

The structure of an SC can be encoded by \emph{boundary operators} $\bB_k$, which record the incidence relations between $(k-1)$-simplices and $k$-simplices according to the chosen reference orientations \cite{Hatcher:2005}.
Rows of $\bB_k$ are indexed by $(k-1)$-simplices and columns of $\bB_k$ are indexed by $k$-simplices.
Thus, the matrix $\bB_1$ is the vertex-to-edge incidence matrix and $\bB_2$ is the edge-to-triangle incidence matrix. 

\subheading{Simplicial Signals}
A $k$-simplicial signal on top of an SC assigns a real number to each oriented $k$-simplex.
We can think of this as a vector in $\bbR^{N_k}$.
For example, $\bbR^{N_0}$ is the \emph{node signal} space and $\bbR^{N_1}$ is the space of \emph{edge flows}. We denote the $k$-simplicial signal space on the SC $X$ by $\cC^k(X)$. Denote by $\cC^\bullet(X)$ the direct sum $\bigoplus_{k=0}^\infty \cC^k(X)$ of the $k$-simplicial signal spaces.
The sign of the signal indicates whether the signal is aligned with the reference orientation or not. More abstractly, a $k$-simplicial signal describes a skew-symmetric function on the set of oriented $k$-simplices. The $k$-th boundary operator $\bB_k$ can be interpreted as a map from $\bbR^{N_k}$ to $\bbR^{N_{k-1}}$. Furthermore, denote by $\del_X\colon \cC^\bullet(X)\to \cC^\bullet(X)$ the direct sum of the maps $\bB_k$.
Because of the way we defined $\bB_k$ and $\del_X$, applying $\del_X$ twice is trivial: $\del_X\circ\del_X=0$.

\subheading{Hodge Laplacian and Hodge Decomposition}
Using the sequence of boundary operators $\bB_k$, we can define the \emph{$k$-th Hodge Laplacian} on the $k$-simplicial signal space $\bbR^{N_k}$ as $\bL_k = \bB_k^\top \bB_k + \bB_{k+1} \bB_{k+1}^\top$ \cite{Horak:2013, Lim:2020}.
The standard graph Laplacian corresponds to $\bL_0 = \bB_1 \bB_1^\top$ ($\bB_0 = 0$ by convention).
The Hodge Laplacian gives rise to the \emph{Hodge decomposition} \cite{Lim:2020, Barbarossa:2020, Schaub:2021, roddenberry21a}:
\begin{align}
	\bbR^{N_k} & = \operatorname{Im}(\bB_{k+1}) \oplus \operatorname{Im}(\bB_k^\top) \oplus \operatorname{ker}(\bL_k).
\end{align}
We denote by $\Delta_X\colon \cC^\bullet(X)\to \cC^\bullet(X)$ the direct sum of the Hodge Laplacians $\bL_k$.

\subheading{Cell Complexes}
Intuitively, \emph{regular cell complexes} (CCs) generalize the notion of SCs \cite{Hatcher:2005,Bredon:1993}.
Rather than allowing only for simplices which are fixed-cardinality subsets of the vertex set as building blocks, we allow for the use of cells with more general shape.
A cell complex $X$ consists of cells of nonnegative dimension.
The \emph{$n$-skeleton} $X_n$ of $X$ consists of all cells in $X$ of dimension at most $n$.
Starting with $0$-cells, a $k$-cell $c_k$ is homeomorphic to the unit ball $B^k := \{\mathbf{x} \in \bbR^k : \|\mathbf{x}\|\le 1 \}$.
Its boundary $S^{k-1} := \{\mathbf{y} \in \bbR^k : \|\mathbf{y}\|= 1 \}$ is glued to $X_{k-1}$ along an attaching map $\alpha\colon S^{k-1} \to X_{k-1}$.
We will call the cells contained in the image of the attaching map $\alpha$ the boundary of $c_k$.
For example, a line is a $1$-cell with its two endpoints as $0$\nohyphen{}dimensional boundary.
Similarly, a polygon is a $2$-cell and its boundary consists of the line segments defining it. Here, the difference between a simplicial and a cell complex becomes clear: While a $2$-cell can be any arbitrary polygon, a $2$-simplex must be a triangle.

Analogous to SCs, we equip a cell complex with a reference orientation for bookkeeping purposes.
A cell inherits its orientation from its boundary.
That is, by orienting the boundary of a cell, we obtain a corresponding orientation of the cell itself.

A CC can again be encoded by boundary operators $\bB_k$. 
The rows of $\bB_k$ are indexed by $(k-1)$-cells and columns of $\bB_k$ are indexed by $k$-cells.
See \cite{Hatcher:2005} for more details. 
If we forget the topological data of a cell complex and only record the incidence and boundary relation, we arrive at the notion of an \emph{abstract cell complex}.
As in the case of SCs, we can associate a $k$-cellular signal space $\cC^k(X)$ to $X$.
We define $\del_X$ analogously to the simplicial complex case.

\section{Temporal Flow Interpolation}
\label{section:motivation}

In this section, we demonstrate how to use product spaces of SCs to process flows varying in time and space.
To that end, we assume that we are given a static simplicial complex and we observe a time-varying edge flow
(\textit{i.e.}, at each time an element of $\cC^1(X)$)
supported on this complex.

\begin{figure}
	\centering

	\resizebox{\linewidth}{!}{\begin{tikzpicture}[>=latex',line join=bevel,very thick, shading=warm]
	\def\xshift{5.5}
	\def\yshift{3.5}


	\node[anchor=center] at (-3.75,0) {\Large{$\mathbf{f},\widehat{\mathbf{f}}$}};
	\node[anchor=center] at (-3.75,-\yshift) {\Large{$\mathbf{f}^*$}};

	\foreach \t in {0,1,2} {
		\node[anchor=south] at (\t*\xshift,1.75) {\Large{$t=\t$}};

		\foreach \r in {0,1} {
			\node (\t0\r) at (\t*\xshift-2.4,-0.7-\r*\yshift) [draw, circle, fill=white] {0};
			\node (\t1\r) at (\t*\xshift-2.4,0.7-\r*\yshift) [draw, circle, fill=white] {1};
			\node (\t2\r) at (\t*\xshift-1,1-\r*\yshift) [draw, circle, fill=white] {2};
			\node (\t3\r) at (\t*\xshift,0-\r*\yshift) [draw, circle, fill=white] {3};
			\node (\t4\r) at (\t*\xshift-1,-1-\r*\yshift) [draw, circle, fill=white] {4};
			\node (\t5\r) at (\t*\xshift+1,1-\r*\yshift) [draw, circle, fill=white] {5};
			\node (\t7\r) at (\t*\xshift+1,-1-\r*\yshift) [draw, circle, fill=white] {7};
			\node (\t6\r) at (\t*\xshift+2,0-\r*\yshift) [draw, circle, fill=white] {6};
		}
	}

	
	\draw[dashed, gray, ->] (000) edge node [black,sloped,anchor=south] {2.25} (040);
	\draw[red!90!lightgray,->] (010) edge node[black,sloped,anchor=south] {2.25} (000);
	\draw[dashed, gray, ->] (020) edge node [black,sloped,anchor=south] {2.25} (010);
	\draw[dashed, gray, ->] (030) edge node [black,sloped,anchor=south] {2.25} (020);
	\draw[dashed, gray, ->] (040) edge node [black,sloped,anchor=south] {1.23} (030);
	\draw[dashed, gray, ->] (040) edge node [black,sloped,anchor=south] {0.92} (070);
	\draw[red!24!lightgray,->] (050) edge node [black,sloped,anchor=south] {0.60} (030);
	\draw[dashed, gray, ->] (060) edge node [black,sloped,anchor=south] {0.60} (050);
	\draw[dashed, gray, ->] (070) edge node [black,sloped,anchor=south] {0.32} (030);
	\draw[dashed, gray, ->] (070) edge node [black,sloped,anchor=south] {0.60} (060);

	\draw[dashed, gray, ->] (100) edge node [black,sloped,anchor=south] {1.96} (140);
	\draw[dashed, gray, ->] (110) edge node[black,sloped,anchor=south] {1.96} (100);
	\draw[dashed, gray, ->] (120) edge node [black,sloped,anchor=south] {1.96} (110);
	\draw[dashed, gray, ->] (130) edge node [black,sloped,anchor=south] {1.96} (120);
	\draw[dashed, gray, ->] (140) edge node [black,sloped,anchor=south] {1.08} (130);
	\draw[dashed, gray, ->] (140) edge node [black,sloped,anchor=south] {0.88} (170);
	\draw[red!28!lightgray,->] (150) edge node [black,sloped,anchor=south] {0.69} (130);
	\draw[dashed, gray, ->] (160) edge node [black,sloped,anchor=south] {0.69} (150);
	\draw[dashed, gray, ->] (170) edge node [black,sloped,anchor=south] {0.19} (130);
	\draw[dashed, gray, ->] (170) edge node [black,sloped,anchor=south] {0.69} (160);

	\draw[red!70!lightgray,->] (200) edge node [black,sloped,anchor=south] {1.76} (240);
	\draw[dashed, gray, ->] (210) edge node[black,sloped,anchor=south] {1.76} (200);
	\draw[dashed, gray, ->] (220) edge node [black,sloped,anchor=south] {1.76} (210);
	\draw[dashed, gray, ->] (230) edge node [black,sloped,anchor=south] {1.76} (220);
	\draw[red!37!lightgray,->] (240) edge node [black,sloped,anchor=south] {0.92} (230);
	\draw[dashed, gray, ->] (240) edge node [black,sloped,anchor=south] {0.85} (270);
	\draw[dashed, gray, ->] (250) edge node [black,sloped,anchor=south] {0.78} (230);
	\draw[dashed, gray, ->] (260) edge node [black,sloped,anchor=south] {0.78} (250);
	\draw[dashed,red!03!lightgray,->] (270) edge node [black,sloped,anchor=south] {0.07} (230);
	\draw[dashed, gray, ->] (270) edge node [black,sloped,anchor=south] {0.78} (260);

	\draw [red!82!lightgray,->] (001) edge node [black,sloped,anchor=south] {2.05} (041);
	\draw [red!82!lightgray,->] (011) edge node[black,sloped,anchor=south] {2.05} (001);
	\draw [red!82!lightgray,->] (021) edge node [black,sloped,anchor=south] {2.05} (011);
	\draw [red!82!lightgray,->] (031) edge node [black,sloped,anchor=south] {2.05} (021);
	\draw [red!46!lightgray,->] (041) edge node [black,sloped,anchor=south] {1.16} (031);
	\draw [red!36!lightgray,->] (041) edge node [black,sloped,anchor=south] {0.89} (071);
	\draw [red!25!lightgray,->] (051) edge node [black,sloped,anchor=south] {0.63} (031);
	\draw [red!25!lightgray,->] (061) edge node [black,sloped,anchor=south] {0.63} (051);
	\draw [red!10!lightgray,->] (071) edge node [black,sloped,anchor=south] {0.26} (031);
	\draw [red!25!lightgray,->] (071) edge node [black,sloped,anchor=south] {0.63} (061);

	\draw [red!77!lightgray,->] (101) edge node [black,sloped,anchor=south] {1.93} (141);
	\draw [red!77!lightgray,->] (111) edge node[black,sloped,anchor=south] {1.93} (101);
	\draw [red!77!lightgray,->] (121) edge node [black,sloped,anchor=south] {1.93} (111);
	\draw [red!77!lightgray,->] (131) edge node [black,sloped,anchor=south] {1.93} (121);
	\draw [red!42!lightgray,->] (141) edge node [black,sloped,anchor=south] {1.06} (131);
	\draw [red!35!lightgray,->] (141) edge node [black,sloped,anchor=south] {0.87} (171);
	\draw [red!27!lightgray,->] (151) edge node [black,sloped,anchor=south] {0.68} (131);
	\draw [red!27!lightgray,->] (161) edge node [black,sloped,anchor=south] {0.68} (151);
	\draw [red!08!lightgray,->] (171) edge node [black,sloped,anchor=south] {0.19} (131);
	\draw [red!27!lightgray,->] (171) edge node [black,sloped,anchor=south] {0.68} (161);

	\draw [red!72!lightgray,->] (201) edge node [black,sloped,anchor=south] {1.81} (241);
	\draw [red!72!lightgray,->] (211) edge node[black,sloped,anchor=south] {1.81} (201);
	\draw [red!72!lightgray,->] (221) edge node [black,sloped,anchor=south] {1.81} (211);
	\draw [red!72!lightgray,->] (231) edge node [black,sloped,anchor=south] {1.81} (221);
	\draw [red!39!lightgray,->] (241) edge node [black,sloped,anchor=south] {0.97} (231);
	\draw [red!33!lightgray,->] (241) edge node [black,sloped,anchor=south] {0.83} (271);
	\draw [red!28!lightgray,->] (251) edge node [black,sloped,anchor=south] {0.69} (231);
	\draw [red!28!lightgray,->] (261) edge node [black,sloped,anchor=south] {0.69} (251);
	\draw [red!06!lightgray,->] (271) edge node [black,sloped,anchor=south] {0.15} (231);
	\draw [red!28!lightgray,->] (271) edge node [black,sloped,anchor=south] {0.69} (261);

	\foreach \t in {0,1,2} {
		\foreach \r in {0,1} {
			\begin{pgfonlayer}{background}
				\fill[fill=lightgray, fill opacity=0.6] (\t3\r.center) to (\t4\r.center) to (\t7\r.center);
			\end{pgfonlayer}
		}
	}

\end{tikzpicture}} \\
	\vspace{10pt}

    \scriptsize
	\begin{NiceTabular}{lrrr}
		\toprule
		$\alpha_t$ & $0.0$ & $0.01$ & $1.0$ \\
		$\alpha_s$ & $1.0$ & $1.0$ & $0.0$ \\
		\midrule
		$\nicefrac{\|\mathbf{f}^*-\mathbf{f}\|_2}{\|\mathbf{f}\|_2}$ & $0.505$ & $\mathbf{0.042}$ & $0.716$ \\
		\bottomrule
	\end{NiceTabular}	  

	\caption{
		\textbf{Illustration: Importance of considering both spatial and temporal components in sample-limited interpolation.}
		The first row illustrates the true flow $\mathbf{f}$ as it varies over time ($t \in \{1,2,3\}$).
		The solid edges indicate the observed flows in each snapshot while dashed edges are unobserved and their flow value should be interpolated.
		We fix $\lambda=10^{-6}$.
		The second row shows $\alpha_t/\alpha_s=0.01$.
	}
    \vspace{-0.5cm}
	\label{figure:demo}
\end{figure}
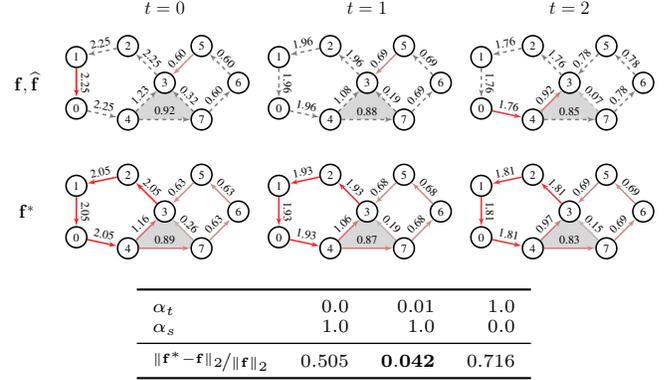

    \vspace{-0.1cm}
\subsection{Example: Interpolation in Time and Space}
    \vspace{-0.1cm}

Let $X=(\mathcal{V},\mathcal{E},\mathcal{T})$ be an SC consisting of vertices, edges, and filled in triangles, and assume that the set of edges has been endowed with an orientation.
Suppose for some time-varying flow $f\colon [T]\times\mathcal{E}\to\bbR$, we observe flows on a subset of the edges within our spatial domain described by $X$.
Denote the collection of observed edges by $\Omega\subseteq [T]\times\mathcal{E}$.
For such a flow, write the vectorized form as $\mathbf{f}$, indexed so that $\mathbf{f}_t^e=f(t,e)$ for some time index $t\in[T]$ and oriented edge $e\in\mathcal{E}$.
For convenience, we write the flow of an edge $e$ at all time steps as $\mathbf{f}^e$, and the flow on the SC at a single time $t$ as $\mathbf{f}_t$.

To interpolate such a partially observed flow, we proceed by setting up the following optimization problem akin to~\cite{Schaub:2018,Jia:2019}, which enforces the flows to vary smoothly (jointly) across space and time:
\begin{equation}\footnotesize
	\label{eq:interpolation-program}
	\begin{aligned}
		\mathbf{f}^* = \argmin_{\mathbf{f}} \bigg(
			& \operatorname{MSE}(\mathbf{f}\big|_\Omega, \widehat{\mathbf{f}}\big|_\Omega) + \alpha_s\sum_{t=0}^T \mathbf{f}_t^\top\bL_s\mathbf{f}_t \\
			& + \alpha_t\sum_{e\in\mathcal{E}} \mathbf{f}^{e\top}\bL_t\mathbf{f}^e + \lambda\mathbf{f}^\top\mathbf{f}
		\bigg).
	\end{aligned}
\end{equation}
Here $\widehat{\mathbf{f}}\big|_\Omega$ denotes the partially observed flow restricted to the observed domain $\Omega$ and $\operatorname{MSE}$ denotes the mean squared error.
Further, $\bL_s$  is the spatial Hodge Laplacian obtained by considering the SC $X$ for each time index $t$ independently. 
Moreover, the temporal Hodge Laplacian $\bL_t$ is here simply the graph Laplacian of the path graph of length $T$, obtained by comparing each edge flow with the flow on the same edge at $t'=t\pm1$.
In this setting, ``smooth'' means ``having a small quadratic form with respect to the Hodge Laplacian.''
The scalars $\alpha_s$ and $\alpha_t$ are regularization parameters, allowing us to enforce varying degrees of temporal smoothness relative to spatial smoothness. 
The ridge penalty $\lambda$ is assumed to be small and is included only to force a unique solution.

As a concrete example, let us assume we observe only part of the edge flow at $T = 3$ different snapshots and want to infer an approximation of the remaining edge flows for all snapshots, as illustrated in~\cref{figure:demo}.
The true (only partially observed) flow is a smooth flow on the edges of an SC and changes gradually over time ($t\in\{1,2,3\}$).
Our observation set now consists of the flow over $2$ edges at $t=1$, $1$ edge at $t=2$, and $2$ edges at $t=3$.

One can clearly see that many edges of the graph are never observed, rendering purely temporal smoothing unsuitable.
Moreover, for any snapshot, the observed edges do not cover a sufficient subset of the graph, so that purely spatial interpolation at each point in time under the assumption of smooth flows is insufficient as well.
This justifies the importance of interpolating in a way that requires joint smoothness in both the temporal and spatial domain.
Indeed, without such smoothness assumptions, the interpolation problem is impossible, since there are $30$ degrees of freedom, but only $5$ known variables.
To test this, we vary the regularization parameters $\alpha_s,\alpha_t$ to explore the relationship between spatial and temporal smoothness, with results shown in the table of \cref{figure:demo}.
Observe that the best interpolation performance is attained when both $\alpha_s$ and $\alpha_t$ are nonzero, with poor performance when either one is zero.

    \vspace{-0.1cm}
\subsection{The Product Space Framework}
    \vspace{-0.1cm}
In this subsection we will show how the above filtering procedure emerges naturally from considering the Hodge Laplacian of the product space.
This not only provides a theoretical justification for considering the above optimization problem, but provides also the means to extend these ideas to signals supported on arbitrary products of complexes.

\subheading{Product Complexes}
Let $X$ and $Y$ be SCs.
Observe that their Cartesian product $Z:=X\times Y$ \emph{cannot} easily be given the structure of a SC.
Instead, we can view $Z$ as a cell complex:
The elements of $Z$ are tuples of the form $(\sigma_X, \sigma_Y)$, where $\sigma_X\in X$ and $\sigma_Y\in Y$.
If $\sigma_X$ is a $k$-simplex, and $\sigma_Y$ is an $\ell$-simplex, then we can view $(\sigma_X,\sigma_Y)$ as a $(k+\ell)$-cell in $Z$.
For instance, a tuple consisting of two vertices yields another vertex, and a tuple consisting of a vertex and an edge yields an edge.
However, a tuple consisting of two edges does not yield a $2$-simplex (triangle), but a $2$-cell resembling a ``filled in rectangle,'' as pictured in~\cref{fig:product}.
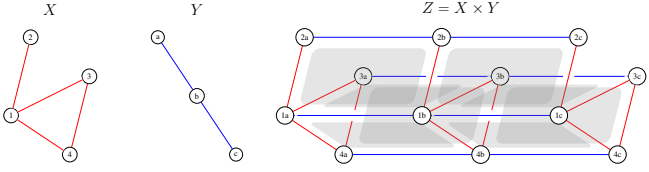
\begin{figure}
	\centering
	\resizebox{\linewidth}{!}{\begin{tikzpicture}[>=latex',line join=bevel,every node/.style={scale=0.5},scale=0.75,shading=coolwarm]
    \def\xshift{3.5}
    
    \tikzset{
    ultra thick/.style={line width=6pt}
    }

    \node[anchor=south] (X) at (-2*\xshift+1,2.5) {\LARGE$X$};
    \node[anchor=south] (Y) at (-\xshift+1.25,2.5) {\LARGE$Y$};
    \node[anchor=south] (Z) at (\xshift+1,2.5) {\LARGE$Z=X\times Y$};

    \node[fill=white] (01) at (-2*\xshift,0) [draw,circle] {1};
    \node[fill=white] (02) at (-2*\xshift+0.5,2) [draw,circle] {2};
    \node[fill=white] (03) at (-2*\xshift+2,1) [draw,circle] {3};
    \node[fill=white] (04) at (-2*\xshift+1.5,-1) [draw,circle] {4};

    \foreach \t/\label in {1/a,2/b,3/c} {
        \begin{pgfonlayer}{ffg}
            \node[fill=white] (\t1) at (\t*\xshift-\xshift,0) [draw,circle] {1\label};
        \end{pgfonlayer}
        
        \begin{pgfonlayer}{fg}
            \node[fill=white] (\t2) at (\t*\xshift-\xshift+0.5,2) [draw,circle] {2\label};
            \node[fill=white] (\t4) at (\t*\xshift-\xshift+1.5,-1) [draw,circle] {4\label};
        \end{pgfonlayer}
        
        \begin{pgfonlayer}{background}
            \node[fill=white] (\t3) at (\t*\xshift-\xshift+2,1) [draw,circle] {3\label};
        \end{pgfonlayer}
    }
    
    \node[fill=white] (p1) at (-\xshift+0.25,2) [draw,circle] {a};
    \node[fill=white] (p2) at (-\xshift+1.25,0.5) [draw,circle] {b};
    \node[fill=white] (p3) at (-\xshift+2.25,-1) [draw,circle] {c};
    
    \draw [blue] (p1) -- (p2) -- (p3);

    \foreach \n/\layer in {1/main,2/main,3/background,4/main} {
        \foreach \t/\nextt in {1/2,2/3} {
            \begin{pgfonlayer}{\layer}
                \draw [white, ultra thick] (\t\n) edge (\nextt\n);
                \draw [blue] (\t\n) edge (\nextt\n);
            \end{pgfonlayer}
        }
    }
    
    \foreach \t in {0,1,2,3} {
        \foreach \ni/\nt/\layer in {1/2/main,1/3/main,1/4/main,3/4/background} {
            \begin{pgfonlayer}{\layer}
                \draw [white, ultra thick] (\t\ni) edge (\t\nt);
                \draw [red] (\t\ni) edge (\t\nt);
            \end{pgfonlayer}
        }
    }
    
    \foreach \t/\nextt in {1/2,2/3} {
        \foreach \ni/\nt/\offset in {1/2/8pt,1/3/6pt,4/1/5pt,4/3/8pt} {
            \def\prev{\t\ni}
            \def\this{\nextt\ni}
            
            \foreach \next in {\nextt\nt,\t\nt,\t\ni,\nextt\ni} {
                \path[name path=first] ($(\prev)!\offset!90:(\this)$) --
                    ($(\this)!\offset!-90:(\prev)$);
                    
                \path[name path=second] ($(\this)!\offset!90:(\next)$) --
                    ($(\next)!\offset!-90:(\this)$);
                    
                \path[name intersections={of=first and second}] coordinate
                    (\prev-\this-\next) at (intersection-1);
                    
                \global\let\prev=\this
                \global\let\this=\next
            }
            
            \begin{pgfonlayer}{fg}
                \fill [fill=gray, fill opacity=0.2, rounded corners]
                    (\nextt\nt-\t\nt-\t\ni.center) -- (\t\nt-\t\ni-\nextt\ni.center) -- 
                    (\t\ni-\nextt\ni-\nextt\nt.center) -- (\nextt\ni-\nextt\nt-\t\nt.center) -- 
                    cycle;
            \end{pgfonlayer}
        }
    }
\end{tikzpicture}}
	\caption{
		\textbf{Cell complex yielded by the Cartesian product of two simplicial complexes.}
		For ease of illustration, we consider two graphs $X,Y$ as the component spaces.
		Taking their Cartesian product, we obtain a cell complex $Z$.
		Using our notation, $\cC^{0,0}(Z)$ is supported on the vertices of $Z$, $\cC^{1,0}(Z)$ is supported on the red edges of $Z$, $\cC^{0,1}(Z)$ is supported on the blue edges of $Z$, and $\cC^{1,1}(Z)$ is supported on the gray faces of $Z$.
	}
    \vspace{-0.5cm}
	\label{fig:product}
\end{figure}

The faces of a $k$-cell $\sigma_Z=(\sigma_X,\sigma_Y)$ are the faces of the form $(\sigma'_X,\sigma_Y)$ or $(\sigma_X,\sigma'_Y)$, where $\sigma'_X$ and $\sigma'_Y$ are faces of $\sigma_X$ or $\sigma_Y$ respectively.
%
%
One can easily see that the product $e\times f$ of two edges  $e=[v_1,v_2]$ and $f=[u_1,u_2]$ has four faces corresponding to the cells $(v_1,f)$, $(v_2, f)$, $(e,u_{1})$, and $(e, u_2)$. 
However, a $2$-\emph{simplex} is only allowed to have $3$ faces. 
To maintain a simplicial structure, we would thus need to subdivide this cell into two triangles.

\subheading{Signal Space of a Product Complex}
Recall that there is a bijective correspondence between $k$-cells of $Z$ and pairs $(\sigma_X^i,\sigma_Y^j)$ of $i$- and $j$-cells of $X$ and $Y$ with $i+j=k$. Thus, we can decompose the $k$-cellular signal space of $Z$ into a direct sum of tensor products of the $i$- and $j$-simplicial signal spaces of $X$ and $Y$.
\begin{equation}
	\cC^k(Z)\cong\bigoplus_{i+j=k}\cC^i(X)\otimes\cC^j(Y).
\end{equation}
We will denote the direct summands by $\cC^{i,j}(Z):=\cC^i(X)\otimes\cC^j(Y)$.
%
A correct definition of the boundary map of $Z$ can be extracted from the boundary maps $\del_X$ and $\del_Y$ as
\begin{equation}
	\begin{aligned}
		\del_Z^{i,j} & \colon\cC^{i,j}(Z)\to \cC^{i-1,j}(Z)\oplus \cC^{i,j-1}(Z),\\
		\del_Z^{i,j} & =\del_X\otimes \id_Y+(-1)^{i}\id_X\otimes \del_Y.
	\end{aligned}
\end{equation}
This implicitly equips the cells of $Z$ with an orientation inherited from the orientations of the simplices in $X$ and $Y$.
The multiplication with $(-1)^i$ is necessary for $\del_Z\circ\del_Z=0$ to hold. 


\subheading{Hodge Laplacian of a Product Space}
With these definitions in place, we can compute the Hodge Laplacian $\Delta_Z$ of the product space $Z:={X}\times {Y}$ as follows.
Recall that we use $X$ to denote the spatial simplicial complex, whereas the temporal simplicial complex $Y$ consists of a path graph with a vertex for each time step $t$ and an edge between each pair of successive time steps $t$ and $t+1$.
For brevity, we will write $\del_{X}$ instead of $\del_{X}\otimes \id_{Y}$ and similarly for $\del_{Y}$, $\Delta_X$, and $\Delta_Y$.
Denoting the restriction of the Hodge Laplacian $\Delta_Z$ to $\cC^{i,j}(Z)$ by $\Delta_Z^{i,j}$ and using the commutativity of $\del_{X}$ with $\del_{Y}^\top$, and $\del_{X}^\top$ with $\del_{Y}$, we obtain the following result.
\begin{prop}
\label{paragraph:HodgeLaplacian}
The Hodge Laplacian $\Delta_Z$ restricted to $\cC^{i,j}(Z)$ is given by
\begin{align}
	\label{eq:productHodge}
	\Delta_Z^{i,j} =\del_Z^{\phantom{\top}}\del_Z^\top+\del_Z^\top\del_Z^{\phantom{\top}}
	=\Delta_{X}+\Delta_{Y},
\end{align}
where $\Delta_{X}$ and $\Delta_{Y}$ are the Hodge Laplacians of the spatial and temporal domain, constructed from $\del_{X}$ and $\del_{Y}$, respectively.
\end{prop}

\begin{proof}
	\begin{align*}
		\Delta_Z^{i,j} & = \del_Z^{\phantom{\top}}\del_Z^\top+\del_Z^\top\del_Z^{\phantom{\top}} \\
		& = \del_{X}^{\phantom{\top}}\del_{X}^\top+(-1)^i\del_{X}^{\phantom{\top}}\del_{Y}^\top+(-1)^{i+1}\del_{Y}^{\phantom{\top}}\del_{X}^\top+\del_{Y}^{\phantom{\top}}\del_{Y}^\top\nonumber \\
		& \quad + \del_{X}^\top\del_{X}^{\phantom{\top}}+(-1)^i\del_{X}^\top\del_{Y}^{\phantom{\top}}+(-1)^{i-1}\del_{Y}^\top\del_{X}^{\phantom{\top}}+\del_{Y}^\top\del_{Y}^{\phantom{\top}}\nonumber \\
		& = \Delta_{X}+\Delta_{Y} \qedhere
	\end{align*}
\end{proof}

Because the above claim holds for arbitrary $i$ and $j$, we see that the Hodge Laplacian acting on $\cC^k(Z)$ simply sums the Hodge Laplacians $\Delta_{X}$ and $\Delta_{Y}$ acting on each constituent vector space of $\cC^k(Z)$.
Thus, it is reasonable to write the Hodge Laplacian on $Z$ as being made up of restricted Laplacians of the form $\Delta_Z^{i,j}\colon\cC^{i,j}(Z)\to\cC^{i,j}(Z)$.
This can be interpreted as an advantage of working with the product structure on a cell complex.
When we are interested in evaluating the Hodge Laplacian only on the $(k+l)$-cells of $Z$ which arose as products of $k$-simplices in ${X}$ and $l$-simplices in ${Y}$, the signal values in $\cC^{k,l}(Z)$ on the $(k+l)$-cells suffice as input.
In contrast, in order to compute the Hodge Laplacian of an arbitrary subset of $n$-cells of $Z$, one needs the entire $(k+l)$-th signal space $\bigoplus_{i+j=k+l}\cC^{i,j}(Z)$ of $Z$ as input.
This simplification yields a high degree of separability when considering the vector spaces $\cC^{i,j}(Z)$, leading to a finer Hodge decomposition.

\subheading{Filtering with the Product Space Hodge Laplacian}
From our above discussion, we can conclude that the interpolation \cref{eq:interpolation-program} for the spatiotemporal flows is simply governed by the Hodge Laplacian of the product space in disguise.
To see this, let again $X$ denote the spatial simplicial complex and $Y$ the temporal simplicial complex.
Furthermore, denote their product cell complex by $Z:=X\times Y$ with signal space $\cC^\bullet(Z)$.
Then the edge flow corresponds to the component $\cC^{1,0}(Z)$ of the signal space.
Using \eqref{eq:productHodge}, we can split the Hodge Laplacian into $\Delta_{X}$ and $\Delta_{Y}$.
Then $\Delta_{X}$ is the spatial Hodge Laplacian $\mathbf{L}_s$ from \eqref{eq:interpolation-program} and $\Delta_{Y}$ the temporal Hodge Laplacian $\mathbf{L}_t$.
We can weight the Hodge Laplacian of $X$ by $\alpha_s$ to represent the spatial smoothness,
and similarly, we can weight the Hodge Laplacian of $Y$ by $\alpha_t$ to represent the temporal smoothness of the signal, so that $\Delta_{Z}^{1,0}=\alpha_s\Delta_{X}^{1}+\alpha_t\Delta_{Y}^{0}$.
This weighting of the Hodge Laplacian yields a quadratic form for temporal flows sufficient to express the interpolation program~\eqref{eq:interpolation-program}, for example.

\section{Real World Example: Ocean Drifters}

In this section, we consider a dataset from the Global Drifter Program, which we restricted to the Caribbean\footnote{Data available at \url{https://www.aoml.noaa.gov/phod/gdp/data.php}. We only consider entries in the rectangle spanned by \ang{25}N\ \ang{90}W (top left) and \ang{10}N\ \ang{55}W (bottom right). We ignore years before 1992 since there is no prior data in this area.}.
The dataset consists of \num{1378} buoys floating in the ocean and the location of each buoy is logged every $6$ hours, resulting in \num{532696} location pings from the years 1992 to 2020.

The spatial simplicial complex is constructed following the setup considered in \cite{Schaub:2020}:
We place a hexagonal grid on the earth's surface, with the size of each hexagon corresponding to \ang{0.3} (latitude).
We associate a vertex with each hexagon and connect two vertices by an edge if their hexagons have a common face.
Every triplet of hexagons that meet at a common point form a filled in triangle.
Hexagons that cover landmasses are removed, which leads to some ``holes'' in the simplicial complex.
The temporal domain corresponds to a path graph of length $29$.

We discretize the observed trajectories according to our simplicial complex.
A trajectory is represented by a vector $\mathbf{f}$ with entry $f_{[i,j]} = 1$ if the edge $[i,j]$ is traversed in the reference orientation, $f_{[i,j]} = -1$ if the edge is traversed in the opposite orientation, and $0$ otherwise.
The flow vector $\mathbf{f}_t$ of a particular year $t = 1, \dots, T$ is the sum of the trajectory vectors of that year, where we split buoy trajectories that span multiple years into separate trajectories.

Finally, we divide the drifters into $80\%$ training and $20\%$ test trajectories and denote with $\smash{\widehat{\mathbf{f}}^\text{tr}}$ and $\smash{\widehat{\mathbf{f}}^\text{tst}}$ the training and test flows, respectively.
We use the training trajectories to estimate the yearly ocean currents in the Caribbean and then check if these currents are compatible with the trajectories of the test drifters, \textit{e.g.}, if the buoy follows the ocean currents or moves perpendicular to them.

More specifically, from $\smash{\widehat{\mathbf{f}}^\text{tr}}$ we want to infer some suitable ocean currents $\mathbf{f}^*$ that describe the given trajectories best.
Denote by $ \mathbf{f}_{\widehat{ \mathbf{f} } }$ the restriction of $\mathbf{f}$ to $\mathrm{supp}\left(\smash{\widehat{\mathbf{f}}}\right)$.
Then this corresponds to minimizing the cosine similarity between $\smash{\widehat{\mathbf{f}}^\text{tr}}$ and $\smash{\mathbf{f}_{\widehat{\mathbf{f}}}}$, \textit{i.e.}, we get the loss function
\begin{equation}
	\mathcal{L}\left( \mathbf{f}, \widehat{\mathbf{f}} \right)
	= \frac{1}{2} \left(1 - \frac{\langle\mathbf{f}, \widehat{\mathbf{f}}\rangle}{\| \mathbf{f}_{\widehat{\mathbf{f}}} \|_2 \cdot \| \widehat{\mathbf{f}} \|_2}\right).
\end{equation}
Notice that we normalize the cosine similarity such that $0$ corresponds to perfect alignment and $1$ to opposite flows.

We further want to infer the flow for hexagons without any training data, \textit{i.e.}, we want to use the resulting flow to infer the test trajectories.
For that we add smoothness assumptions both in space and in time to the objective, \textit{i.e.}, we penalize if a hexagon's flow differs significantly from its neighbors' flows and if a hexagon's flow varies largely over time.
The resulting convex optimization problem is as follows:
\begin{equation}
    \scriptsize
	\mathbf{f}^* = \argmin_{\mathbf{f}} \left(
		\mathcal{L}\left( \mathbf{f},  \widehat{\mathbf{f}}^\text{tr} \right)
		+ \frac{
			\frac{\alpha_s}{T} \sum_t \mathbf{f}_t^\top\bL_s\mathbf{f}_t
			+ \frac{\alpha_t}{|\mathcal{E}|} \sum_e \mathbf{f}^{e\top}\bL_t\mathbf{f}^e
		}{\| \mathbf{f} \|_2^2}
	\right).
\end{equation}
This amounts to minimizing $
\mathcal{L}(\mathbf{f},\widehat{\mathbf{f}}^\text{tr}) +
\langle\mathbf{f},\Delta(\alpha_s,\alpha_t)\mathbf{f}\rangle/\|\mathbf{f}\|_2^2$, where $\Delta (\alpha_s,\alpha_t)=\alpha_s\Delta_X/T+\alpha_t\Delta_Y/|\mathcal{E}|$ denotes the weighted Hodge Laplacian.
We solve this optimization problem for hyperparameters $\alpha_t, \alpha_s \in \{0\} \cup \{10^i : i = 0, \dots, -5 \}$ and report the losses in the table of \cref{fig:ocean-drifter-losses}.
Almost all hyperparameters attain low loss on the training data, especially with those that give higher weight to spatial smoothing.
For the test data, we observe that $\alpha_t = 1.0$ and $\alpha_s =\smash{10^{-3}}$ perform best.
That is, enforcing more temporal smoothing than spatial smoothing is beneficial to infer unseen buoy trajectories.
We assess that this is the case because ocean currents have fewer variations over the years as they have spatially.

\begin{figure}
	\centering

	\includegraphics[width=0.9\linewidth]{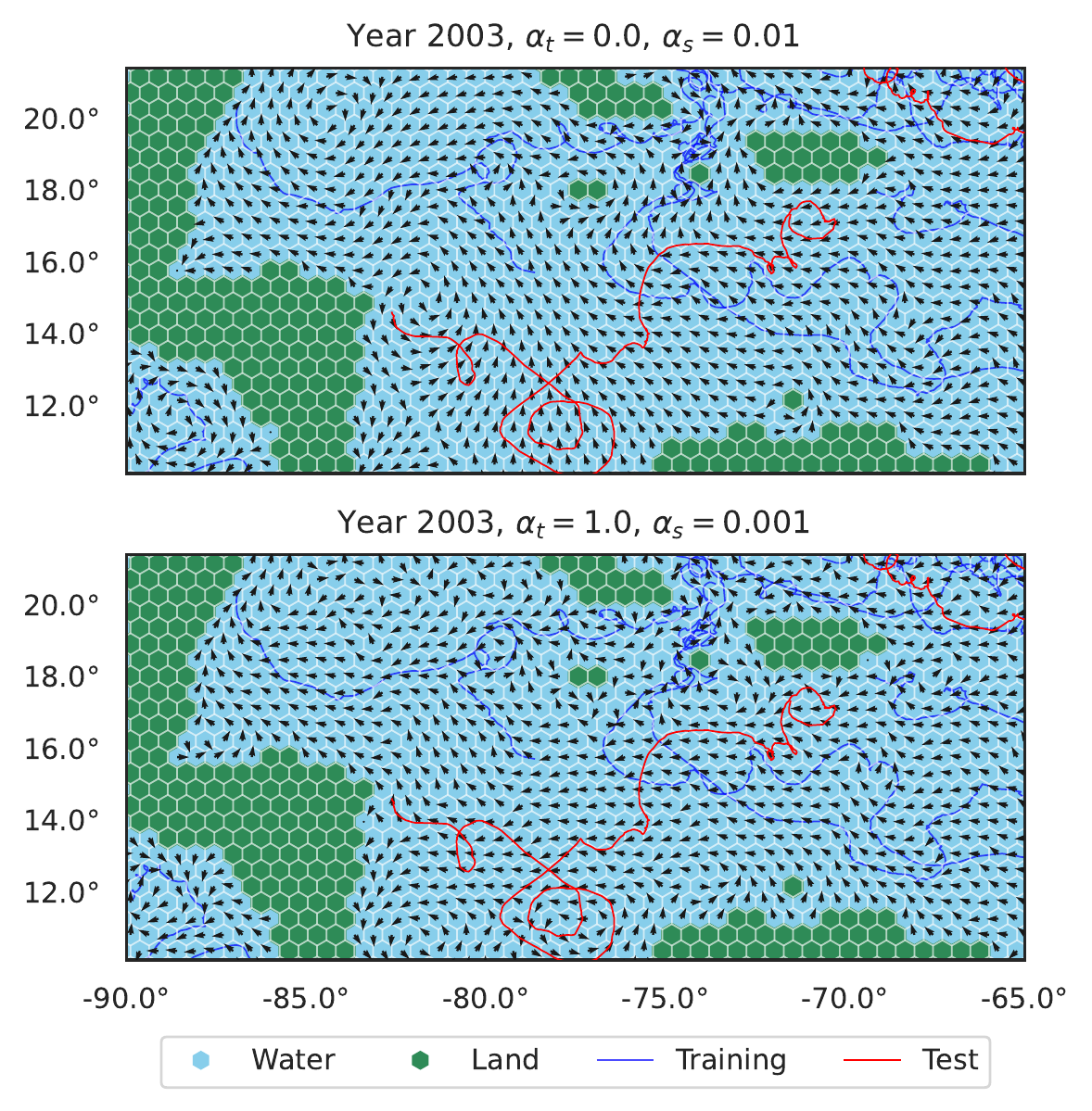}

    \scriptsize
	\begin{tabular}{lrrrr}
          \toprule
		  $\alpha_t$ & $0.0$ & $1.0$ & $1.0$ & $0.0$ \\
		  $\alpha_s$ & $0.0$ & $0.0$ & $\num{e-3}$ & $\num{e-2}$ \\
          \midrule
          $\mathcal{L}(\mathbf{f}^*, \widehat{\mathbf{f}}^\text{tr})$ & $\mathbf{0.000}$ & $\mathbf{0.000}$ & $\mathbf{0.000}$ & $\mathbf{0.000}$ \\
		  $\mathcal{L}(\mathbf{f}^*, \widehat{\mathbf{f}}^\text{tst})$ & $0.316$ & $0.290$ & $\mathbf{0.264}$ & $0.303$ \\
          \bottomrule
	\end{tabular}

	\caption{
		\textbf{Quality of the inferred ocean currents with different smoothness hyperparameters.}
		Shown are the best cases without smoothness assumptions (first column), with pure temporal and pure spatial smoothness constraints (second and last column) and a combination of both (third).
		A loss of $0$ corresponds to perfect alignment while $1$ indicates opposite flows.
	}
    \vspace{-0.5cm}
	\label{fig:ocean-drifter-losses}
\end{figure}


\section{Discussion}

We have shown that the Hodge Laplacian of product cell complexes provides a concise and effective framework for filtering spatiotemporal data.
We showed the effectiveness of this product space framework by extrapolating ocean flows in the Caribbean from a limited set of data points varying in time and space.
Further research could be made using more general product spaces. For instance, one could model a time complex where each time step is connected to the next time step and the time one year later.
This approach could harvest periodicities on multiple scales in the data.
This is connected to using a fiber bundle structure on the time space.


\bibliographystyle{IEEEbib}
\bibliography{IEEEabrv,ref}

\end{document}